\pdfoutput=1

\documentclass{sig-alt-full}

\usepackage{latexsym}
\usepackage{url}
\usepackage{xspace}
\usepackage{comment}
\usepackage{amsmath,amssymb}

\newtheorem{theorem}{Theorem} %[section]
\newtheorem{lemma}[theorem]{Lemma}

\newtheorem{definitionAux}[theorem]{Definition}
\newenvironment{definition}{\begin{definitionAux}\rm}{\end{definitionAux}}

\newtheorem{claimAux}{Claim}

\newtheorem{exampleAux}{Example}

\def\qed{\hfill{\qedboxempty}      % qed with empty box
  \ifdim\lastskip<\medskipamount \removelastskip\penalty55\medskip\fi}

\def\qedboxempty{\vbox{\hrule\hbox{\vrule\kern3pt
                 \vbox{\kern3pt\kern3pt}\kern3pt\vrule}\hrule}}

\def\qedfull{\hfill{\qedboxfull}   % qed with full box
  \ifdim\lastskip<\medskipamount \removelastskip\penalty55\medskip\fi}

\def\qedboxfull{\vrule height 4pt width 4pt depth 0pt}

\newcommand{{\incolumn}}[1]{\begin{tabular}[c]{c} #1 \end{tabular}}
\newcommand{{\incolumnmath}}[1]{\begin{array}[c]{c} #1 \end{array}}
\newcommand{\A}{\mathcal{A}} 
 \newcommand{\D}{\mathcal{D}}

 \renewcommand{\L}{\mathcal{L}}
 
 \renewcommand{\P}{\mathcal{P}}
\newcommand{\Q}{\mathcal{Q}} \newcommand{\R}{\mathcal{R}}
\renewcommand{\S}{\mathcal{S}} 
 \newcommand{\V}{\mathcal{V}}

\newcommand{\ra}{\rightarrow}

\newcommand{\EXPTIME}{\textsc{ExpTime}\xspace}
\newcommand{\EXPSPACE}{\textsc{ExpSpace}\xspace}
\newcommand{\NEXPSPACE}{\textsc{NExpSpace}\xspace}
\newcommand{\NEXPTIME}{\textsc{NExpTime}\xspace}

\newcommand{\Qs}{\Q_s}
\newcommand{\Qt}{\Q_t}
\newcommand{\Qv}{\Q_v}

\newcommand{\As}{\Sigma_s}
\newcommand{\At}{\Sigma_t}

\newcommand{\Ds}{\D_s}
\newcommand{\Dt}{\D_t}

\renewcommand{\aa}{a}
\newcommand{\bb}{b}
\newcommand{\pp}{p}

\newcommand{\map}[2]{#1\leadsto #2}

\newcommand{\edge}[3]{(#1,#2,#3)}

\title{View Synthesis from Schema Mappings}

\numberofauthors{3}

\author{
 \alignauthor Diego Calvanese\\[1mm]
   \affaddr{Faculty of Computer Science}\\
   \affaddr{Free Univ.\ of Bozen-Bolzano}\\
   \affaddr{I-39100 Bolzano, Italy}\\
   \email{calvanese@inf.unibz.it}
 \alignauthor Giuseppe De Giacomo\\Maurizio Lenzerini\\[1mm]
   \affaddr{Dip.\ di Inf.\ e Sist.}\\
   \affaddr{Univ.\ di Roma ``La Sapienza''}\\
   \affaddr{I-00198 Roma, Italy}\\
   \email{lastname@dis.uniroma1.it}
 \alignauthor Moshe Y.~Vardi\\[1mm]
   \affaddr{Dep.\ of Computer Science}\\
   \affaddr{Rice University, P.O.~Box 1892}\\
   \affaddr{Houston, TX 77251-1892, U.S.A.}\\
   \email{vardi@cs.rice.edu}
}

\permission{}
\copyrightetc{}
\begin{document}

\maketitle

\begin{abstract}
  In data management, and in particular in data integration, data exchange,
  query optimization, and data privacy, the notion of view plays a central
  role.  In several contexts, such as data integration, data mashups, and data
  warehousing, the need arises of designing views starting from a set of known
  correspondences between queries over different schemas.  In this paper we
  deal with the issue of automating such a design process.  We call this novel
  problem ``view synthesis from schema mappings'': given a set of schema
  mappings, each relating a query over a source schema to a query over a target
  schema, automatically synthesize for each source a view over the target
  schema in such a way that for each mapping, the query over the source is a
  rewriting of the query over the target wrt the synthesized views.  We study
  view synthesis from schema mappings both in the relational setting, where
  queries and views are (unions of) conjunctive queries, and in the
  semistructured data setting, where queries and views are (two-way) regular
  path queries, as well as unions of conjunctions thereof.  We provide
  techniques and complexity upper bounds for each of these cases.
\end{abstract}

\section{Introduction}
\label{sec:introduction}

A \emph{view} is essentially a (virtual or materialized) data set
that is known to be the the result of executing a specific query over
an underlying database. There are several data-management tasks where
the notion of view plays an important role~\cite{Hale00}.
\begin{itemize}
\item In database design, following the well-known principle of data
  independence, views may be used to provide a logical description of
  the storage schema (cf., \cite{TsSI96}). In this setting,
  since queries are expressed at the logical level, computing a query
  plan over the physical storage involves deciding how to use the
  views in the query-answering process.
\item
In query optimization~\cite{ChCS09}, the computation of the answer set
to a query may take advantage of materialized views, because part of
the data needed for the computation may be already available in the view
extensions.
\item
In data privacy, authorization views associated with a user represent
the data that such user is allowed to access~\cite{ZhMe05}. When the
system computes the result of a query posed to a specific user, only
those answers deriving from the content of the corresponding
authorization views are provided to the user.
\item
In data integration, data warehousing, and data exchange, a target
schema represents the information model used to either accessing, or
materializing the data residing in a set of
sources~\cite{Kola05,Lenz02}. In these contexts, views are used to
provide a characterization of the semantics of the data sources in
terms of the elements of the target schema, and answering target
queries amounts to suitable accessing the views.
\end{itemize}

The above discussion points out that techniques for using the
available views when computing the answers to query are needed in a
variety of data management scenarios. Query processing using views
is defined as the problem of computing the answer to a query by
relying on the knowledge about a set of views, where by ``knowledge''
we mean both view definitions and view extensions~\cite{Hale01}.

\smallskip {\bf View-based query processing.}  Not surprisingly, the
recent database literature witnesses a proliferation of methods,
algorithms and complexity characterizations for this problem. Two
approaches have emerged, namely, query rewriting and query
answering. In the former approach, the goal is to reformulate the
query into an expression that refers to the views (or only to the
views), and provides the answer to the query when evaluated over the
view extension. In the latter approach, one aims at computing the
so-called certain answers, i.e., the tuples satisfying the query in
all databases consistent with the views.

Query rewriting has been studied in relational databases for the case
of conjunctive queries, and many of their variants, both with and
without integrity constraints (see a survey in~\cite{Hale01}). A
comprehensive framework for view-based query answering in relational
databases, as well as several interesting complexity results for
different query languages, are presented
in~\cite{AbDu98,GrMe99}.
%%
%%% LONGER VERSION WITH MORE CITATIONS
%% Query rewriting has been studied for the case of conjunctive queries (with
%% or without arithmetic comparisons)~\cite{LMSS95,RaSU95,PoLe00}, disjunctive
%% views~\cite{AfGK99}, queries with aggregates~\cite{SDJL96,CoNS99}, recursive
%% queries and nonrecursive views~\cite{DuGe97}, queries with negated
%% goals~\cite{FlGr01}, and in the presence of integrity
%% constraints~\cite{Gryz98,DuLe97,CaLR03b} and of limitations in accessing the
%% views~\cite{LiCh01,CaCM09}. Rewriting techniques for query optimization are
%% described, for example, in~\cite{CKPS95,ACPS96,TsSI96}. A comprehensive
%% framework for view-based query answering in relational databases, as well as
%% several interesting results, are presented in~\cite{GrMe99}.
%% In~\cite{AbDu98}, an analysis of the data complexity of the problem under
%% various assumptions is carried out for the case where the views and the
%% queries are expressed in terms of various languages (conjunctive queries,
%% Datalog, first-order queries, etc.). Further results in the presence of
%% integrity constraints are reported in~\cite{CaLR03b}, and in the context of
%% data exchange and Description Logics in~\cite{Kola05,Libk06}, and
%% \cite{CDLR08}, respectively.

View-based query processing has also been
addressed in the context of semi-structured databases. In the case of
graph-based models, the problem has been studied for the class of
regular path queries and its extensions (see, for
example,~\cite{CDLV02c,GrTh03a}. In the case of XML-based model, results on
both view-based query rewriting and view-based query answering are
reported in for several variants of the XPath query language (see, for
example,~\cite{ACGK*09,CDLV09}.
%%
%%% LONGER VERSION WITH MORE CITATIONS
%% View-based query processing has also been addressed in the context of
%% semi-structured databases. In the case of graph-based models, the problem
%% has been studied for the class of regular path queries in
%% \cite{CDLV99b,CDLV02c,GrTh03,GrTh03a}. In~\cite{CDLV00,CDLV00c}, a further
%% distinction is proposed for characterizing the domain of the database (open
%% vs.\ closed domain assumption), and the problem is studied for the case of
%% regular-path queries, both with and without the inverse operator. In the
%% case of XML-based model, results on both view-based query rewriting and
%% view-based query answering are reported
%% in~\cite{FeSu98,MiSu99,PaVa99,ACGK*09,CDLV09} for several variants of the
%% XPath query language.

\smallskip {\bf Where do the views come from?}
All the above works assume that the set of views to be used during query
processing is available. Therefore, a natural question arises: where do these
views come from? Some recent papers address this issue from different points
of views. In~\cite{DPGW10}, the authors introduce the so-called ``view
definition problem'': given a database instance and a corresponding
view instance, find the most \emph{succinct} and accurate view definition, for
a specific view definition language. Algorithms and complexity results
are reported for several family of languages. (Note that the problem dealt
with in~\cite{TrCP09} can be seen as a variant of the view definition
problem.)

In the context of both query optimization and data warehousing, there
has been a lot of interest in the so-called ``view-selection
problem''~\cite{ChHS01}, that is the problem of choosing a set of
views to materialize over a database schema, such that the cost of
evaluating a set of workload queries is minimized and such that the
views fit into a pre-specified storage constraint. Note that the input
to an instance of this problem includes knowledge about both a set of
queries that the selected views should support, and a set of
constraints on space limits for the views.

In data integration and exchange, the ``mapping discovery problem''
has received significant attention in the last years: find
correspondences between a set of data sources and a target (or,
global) schema so that queries posed to the target can be answered by
exploiting such mappings, and accessing the sources
accordingly. Several types of mappings have been investigated in the
literature~\cite{Lenz02}. In particular, in the so-called LAV
(Local-As-Views) approach, mappings associate to each source a view
over the target schema. In other words, the LAV approach to data
integration and exchange advocate the idea of modeling each source as
a view.

The problem of semi-automatically discovering mappings has been
addressed both by the database and AI
communities~\cite{RaBe01,GiYS07}. In~\cite{SeGo08}, a theoretical
framework is presented for discovering relationships between two
database instances over distinct schemata. In particular, the problem
of understanding the relationship between two instances is formalized
as that of obtaining a schema mapping so that a minimum repair of this
mapping provides a perfect description of the target
instance. In~\cite{DLDH04}, the iMAP system is described, which
semi-automatically discovers both 1-1 and complex matches between different
data schemata, where a match specifies semantic correspondences between
elements of both schemas, and is therefore analogous to mappings. None
of the above papers addresses the issue of automatically deriving LAV
mappings. This implies that none of the methods described in those
papers can be used directly to derive the view definitions associated
with the data sources.

\smallskip {\bf Synthesizing views from schema mappings.}  In this
paper, we tackle the problem of deriving view definitions from a
different angle. We assume that we have as input a set of schema
mappings, i.e., a set of correspondences between a source schema and a
target schema, where each correspondence relates a source query (i.e.,
a query over the sources) to a target query. The goal is to
automatically synthesize one view for each source relation, in such a way
that all schema mappings are captured. We use two interpretations of a
``schema mapping captured by the synthesized views''. Under the former
interpretation, the schema mapping is captured if the source query of
such mapping is a nonempty, sound rewriting of the target query with
respect to the views. Under the latter interpretation, the mapping is
captured if the source query is an exact rewriting of the target query
with respect to the views. We remind the reader that, given a set of
views $V$, a query $q_v$ over the set of view symbols in $V$ is called
a sound (exact) rewriting of a target query $q_t$ with respect to $V$
if, for each target database that is coherent with the extensions of
views $V$, the result of evaluating $q_v$ over the view extensions is a
subset of (equal to) the result of evaluating $q_t$ over the target
database.

We call this problem \emph{(exact) view synthesis from schema
  mappings}. We also refer to the decision problem associated to view
synthesis, called \emph{(exact) view existence}: check whether there
exists a set of views, one for each source, that captures all the schema
mappings.

The view-synthesis problem is relevant in several scenarios. We
briefly discuss some of them.
\begin{itemize}
\item In data warehousing, based on the consideration that business
  value can be returned as quickly as the first data marts can be
  created, the project often starts with the design of a few data
  marts, rather than with the design of the complete data warehouse
  schema. Designing a data mart involves deciding how data extracted
  from the sources populate the data warehouse concepts that are
  relevant for that data mart~\cite{Inmo96}.  In this context,
  view synthesis amounts to derive, from a set of specific data marts
  already defined, a set of LAV mappings from the data sources to the
  elements of the data warehouse. With such mappings at hand, the
  design of further data marts is greatly simplified: it is sufficient
  to characterize the content of the new data mart in terms of a query
  over the virtual warehouse, and the extraction program will be
  automatically derived by rewriting the query with respect to the
  synthesized views.

\item Similar to the case described above, real-world
  information-integration projects start by designing wrappers, i.e.,
  processes that extract data from the sources and provide single services
  for the user. This is typically the scenario of portal design, where
  data integration is performed on a query-by-query basis. Each query
  is wrapped to a service, and each time this service is invoked
  through the portal, the extraction program is activated, and the
  specific data integration task associated to it is performed. A much
  more modular, extensible, and reusable architecture is the one where
  a full-fledged data integration system, comprising the global (or,
  target) schema and the mapping to the sources, replaces this
  query-by-query architecture. View synthesis provides the technique
  to automatically derive such a data integration system. Indeed, if
  the various services are characterized in terms of queries over a
  target alphabet, the combination of wrappers and the corresponding
  queries over the target form a set of schema mappings, from which
  the view synthesis algorithm produces the LAV mappings the
  constitute the data integration system.

\item Recently, there has been some interest in so-called data
  mashups. A mashup is a web application that combines data or
  functionality from a collection of external sources, to create a new
  information service~\cite{DHPB09}. Describing the semantics of
  such a service means to describe it as a query over a
  domain-specific alphabet. Once this has been done, the mashup is
  essentially characterized as a schema mapping from the external
  sources to a virtual global database. So, similarly to the above
  mentioned cases, view synthesis can be used to turn the set of
  mashups into a full-fledged data LAV data integration system, with
  all the advantages pointed out before.
\end{itemize}

In all the above scenarios, view-synthesis is used for deriving a set
of LAV mappings starting from a set of available schema mappings. This
is not surprising, since, as we said before, in the LAV approach
sources are modeled as views. Nevertheless, one might wonder why
deriving the LAV mappings, and not using directly the original schema
mappings for data warehousing, integration and mashup. There are
several reasons why one is interested in LAV mappings:
\begin{itemize}
\item LAV mappings allow one to exploit the algorithms and the
  techniques that have been developed for view-based query processing
  in the last years.
\item Several recent papers point out that the language of LAV
  mappings enjoys many desirable properties. For example, in~\cite{tCKo09} it
  is shown that LAV mappings always admit universal
  solutions, allow the rewriting of unions of conjunctive queries over
  the target into unions of conjunctive queries over the sources, and
  are closed both under target homomorphism and union. Recently, LAV
  mappings have also been shown to be closed under composition, and to
  admit polynomial time recoverability checking~\cite{ArFM10}.
\item LAV mappings allow a characterization of the sources in terms of
  the element of the target schema, and, therefore, are crucial in all
  the scenarios where a precise understanding, and a formal
  documentation of the content of the sources are needed.
\end{itemize}

\smallskip {\bf Contributions of the paper.} In this paper we propose
a formal definition of the view-synthesis and the view-existence
problems, and present the first study on such problems, both in the
context of the relational model, and in the context of semistructured
data.

For relational data, we address the case where queries and views
are both conjunctive queries, and the case where they are unions of
conjunctive queries. In the former case, we show that both
view-existence and exact view-existence are in NP. In the latter case,
we show that both problems are in $\Pi^p_2$.

In the context of semistructured data, we refer to a graph-based data
model, as opposed to the popular XML-based model. The reason is that
in many interesting scenarios, including the ones where XML data are
used with \textsf{refids}, semistructured data form a graph rather than a
tree. For graph-based semistructured data, we first study
view synthesis and view existence in the cases where queries and views
are regular path queries. We first present a techniques for
view-existence based on automata on infinite trees~\cite{VaWo86}, and
provide an \EXPTIME upper bound for the problem. We then illustrate an
alternative technique based on the characterization of regular languages
by means of left-right congruence classes. Such a characterization allows
us to prove an \EXPSPACE upper bound for the exact view-existence problem.
Finally, by exploiting a language-theoretic characterization for
containment of regular path queries with inverse (called two-way regular
path queries) provided in~\cite{CDLV03c}, we extend the congruence
class-based technique to the case where queries and views are two-way
regular path queries, as well as conjunctive two-way regular path
queries, and unions of such queries.

\smallskip {\bf Organization of the paper.}  The paper is organized as
follows. In Section~\ref{sec:preliminaries}, we recall some preliminary
notions. In Section~\ref{sec:view-synthesis}, we formally define the problem of
view-synthesis from schema mappings, and the problem of view-existence. In
Section~\ref{sec:UCQs}, we study the problem in the case where queries and
views are conjunctive queries, and unions thereof. In
Section~\ref{sec:tree-solution} and Section~\ref{sec:cc-solution}, we
illustrate the techniques for the view synthesis problem in the case of RPQs
over semistructured data. Section~\ref{sec:extensions} extends the technique to
two-way RPQs, and to (unions of) conjunctive two-way RPQs,
respectively. Section~\ref{sec:conclusions} concludes the paper.

%%% Local Variables:
%%% mode: latex
%%% TeX-master: "main"
%%% save-place: t
%%% End:

%%%%%%%%%%%%%%%%%%%%%%%%%%%%%%%%%%%%%%%%%%%%%%%%%%%%%%%%%%%%%%%%%%%%%%%%%%%%%%
%%% Time-stamp: "2009-12-14 00:36:09 calvanese"
%%%%%%%%%%%%%%%%%%%%%%%%%%%%%%%%%%%%%%%%%%%%%%%%%%%%%%%%%%%%%%%%%%%%%%%%%%%%%%

%%%%%%%%%%%%%%%%%%%%%%%%%%%%%%%%%%%%%%%%%%%%%%%%%%%%%%%%%%%%%%%%%%%%%%%%%%%%%%
\section{Preliminaries}
\label{sec:preliminaries}
%%%%%%%%%%%%%%%%%%%%%%%%%%%%%%%%%%%%%%%%%%%%%%%%%%%%%%%%%%%%%%%%%%%%%%%%%%%%%%

In this work we deal with two data models, the standard relational
model~\cite{AbHV95}, and the graph-based semistructured data
model~\cite{CDLV02c}.

Given a (relational) alphabet $\Sigma$, a database $\D$ over $\Sigma$, and a
query $q$ over $\Sigma$, we denote with $q^{\D}$ the set of tuples resulting
from evaluating $q$ in $\D$.  A query $q$ over $\Sigma$ is \emph{empty} if for
each database $\D$ over $\Sigma$ we have $q^{\D}=\emptyset$.  Given two queries
$q_1$ and $q_2$ over $\Sigma$, we say that $q_1$ is \emph{contained in} $q_2$,
denoted $q_1\sqsubseteq q_2$, if $q_1^{\D}\subseteq q_2^{\D}$ for every
database $\D$ over $\Sigma$.  The queries $q_1$ and $q_2$ are
\emph{equivalent}, denoted $q_1\equiv q_2$, if both $q_1\sqsubseteq q_2$ and
$q_2\sqsubseteq q_1$.

We assume familiarity with the relational model and with (unions of)
conjunctive queries, (U)CQs, over a relational database.  Below we recall the
basic notions regarding the graph-based semistructured data model and regular
path queries.

A \emph{semistructured database} is a finite graph whose nodes represent
objects and whose edges are labeled by elements from an alphabet of binary
relational symbols~\cite{Bune97,Abit97,CDLV03c}.  An edge $\edge{o_1}{r}{o_2}$
from object $o_1$ to object $o_2$ labeled by $r$ represents the fact that
relation $r$ holds between $o_1$ and $o_2$.
%%
%% Note that a semistructured database can be viewed as a structure $\B$ over
%% the set of binary relational symbols. % with $\dom{\B}\subseteq\Delta$.
%%
A \emph{regular-path query} (RPQ) over an alphabet $\Sigma$ of binary relation
symbols is expressed as a regular expression or a \emph{nondeterministic finite
 word automaton (NWA)} over $\Sigma$.  When evaluated on a (semistructured)
database $\D$ over $\Sigma$, an RPQ $q$ computes the set $q^{\D}$ of pairs of
objects connected in $\D$ by a path in the regular language $\L(q)$ defined by
$q$.  Containment between RPQs can be characterized in terms of containment
between the corresponding regular languages:
%%\begin{theorem}[\cite{CDLV02c}]
%%  \label{thm:rpq-containment}
given two RPQs $q_1$ and $q_2$, we have that $q_1\sqsubseteq q_2$ iff
$\L(q_1)\subseteq\L(q_2)$ \cite{CDLV02c}.
%%\end{theorem}

We consider also \emph{two-way regular-path queries}
(2RPQs)~\cite{CDLV00c,CDLV03c}, which extend RPQs with the \emph{inverse}
operator.  Formally, let $\Sigma^\pm=\Sigma\cup\{r^- \mid r\in\Sigma\}$ be the
alphabet including a new symbol $r^-$ for each $r$ in $\Sigma$.  Intuitively,
$r^-$ denotes the inverse of the binary relation $r$.  If $p\in\Sigma^\pm$,
then we use $p^-$ to mean the \emph{inverse} of $p$, i.e., if $p$ is $r$, then
$p^-$ is $p^-$, and if $p$ is $r^-$, then $p^-$ is $r$.
2RPQs are expressed by means of an NWA over $\Sigma^\pm$.
%%% DIEGO removed: 13/12/2009
%% whose language is different from the language consisting only of the empty
%% word $\varepsilon$.
%%% END removed
%%
When evaluated on a database $\D$ over $\Sigma$, a 2RPQ $q$ computes
the set $q^{\D}$ of pairs of objects connected in $\D$ by a semipath that
conforms to the regular language $\L(q)$.
A \emph{semipath} in $\D$ from $x$ to $y$ (labeled with $p_1\cdots p_n$) is a
sequence of the form
$(y_0,p_1,y_1,\ldots,y_{n-1},p_n,y_n)$,
where $n\geq 0$, $y_0=x$, $y_n=y$, and for each $y_{i-1},p_i,y_i$, we have that
$p_i\in\Sigma^\pm$, and, if $p_i=r$ then $(y_{i-1},y_i)\in r^{\D}$, and if
$p_i=r^-$ then $(y_i,y_{i-1})\in r^{\D}$.
We say that a semipath $(y_0,p_1,\ldots,p_n,y_n)$ \emph{conforms to} $q$ if
$p_1\cdots p_n\in\L(q)$.
%% A semipath is said to be \emph{simple} if no object appears more than once
%% in the corresponding sequence.

We will also consider conjunctions of 2RPQs and their unions, abbreviated
(U)C2RPQs~\cite{CDLV00b}, which are (unions of) conjunctive queries constituted
only by binary atoms whose predicate is a 2RPQ.
Specifically, a C2RPQ $q$ of arity $n$ is written in the form
\[
  q(x_1,\ldots,x_n) \leftarrow
    q_1(y_1,y_2) \land\cdots\land q_m(y_{2m-1},y_{2m})
\]
where $x_1,\ldots,x_n,y_1,\ldots,y_{2m}$ range over a set $\{z_1,\ldots,z_k\}$
of variables, $\{x_1,\ldots,x_n\}\subseteq\{y_1,\ldots,y_{2m}\}$, and each
$q_j$ is a 2RPQ.  When evaluated over a database $\D$ over $\Sigma$, the C2RPQ
$q$ computes the set of tuples $(o_1,\ldots,o_n)$ of objects such that there is
a total mapping $\varphi$ from $\{z_1,\ldots,z_k\}$ to the objects in $\D$ with
$\varphi(x_i)=o_i$, for $i\in\{1,\ldots,n\}$, and
$(\varphi(y_{2j-1}),\varphi(y_{2j}))\in q_j^{\D}$, for $j\in\{1,\ldots,m\}$.

Containment between 2RPQs and (U)C2RPQs can also be characterized in terms 
of containment between regular languages.  We elaborate on this in
Section~\ref{sec:extensions}.
%%% DIEGO removed: 13/12/2009
%% However, in this case we have to resort to the notion of \emph{folding} of a
%% language~\cite{CDLV03c}, which intuitively denotes the set of words that are
%% the result of repeatedly cancelling out adjacent occurrences of a symbol and
%% its inverse.  Specifically, given two 2RPQs $q_1$ and $q_2$, we have that
%% $q_1\sqsubseteq q_2$ iff $\L(q_1)\subseteq\mathit{fold}(\L(q_2))$, where
%% $\mathit{fold}(\L)$ denotes the language obtained from $\L$ by folding the
%% words in $\L$~\cite{CDLV03c}.
%%% END removed
%%
We conclude by observing that (U)CQs, RPQs, 2RPQs, and (U)C2RPQs are monotone.
%% and that every set of 2RPQ (resp., RPQ) views is non-constraining.

%%% Local Variables:
%%% mode: latex
%%% TeX-master: "main"
%%% save-place: t
%%% End:

\section{The View-Synthesis Problem}
\label{sec:view-synthesis}

The view-synthesis and the view-existence problems refer to a scenario
with one source schema, one target schema, and a set of schema
mappings between the two, where the goal is to synthesize one view for
each source.

To model the source and the target schemas we refer to two finite
alphabets, the \emph{source alphabet} $\As$ and the \emph{target
  alphabet} $\At$, and to model the queries used in both the mappings
and the views, we use three query languages, namely, the \emph{source
  language} $\Qs$ over $\As\cup\At$, the \emph{target language} $\Qt$
over $\At$, and the \emph{view language} $\Qv$ over $\At$. Notice that
queries expressed in the language $\Qs$ may use symbols in the target
alphabet.

A \emph{schema mapping}, or simply a \emph{mapping}, between the
source and the target is a statement of the form $\map{q_s}{q_t}$,
with $q_s\in\Qs$ and $q_t\in\Qt$. Intuitively, a mapping of this type
specifies that all answers computed by executing the source query
$q_s$ are answers to the target query $q_t$. This means that
$q_s\in\Qs$ is actually a rewriting of $q_t$. This explains why we allow $\Qs$
to use symbols in the target alphabet: in general, the rewriting of a
target query may use symbols both in the source alphabet, and in the
target alphabet~\cite{LMSS95}.

The problem we consider aims at defining one view for each source,
in such a way that all input schema mappings are captured. The
\emph{views} $V$ over $\At$ to be synthesized are modeled as a (not
necessarily total) function $V:\As\rightarrow\Qv$ that associates to
each source symbol $\aa\in\As$ a query $V(\aa)\in\Qv$ over the target
alphabet $\At$. As we said before, our notion of ``views capturing a
set of mappings'' relies on view-based query rewriting, whose
definition we now recall. In the following, given a source database
$\Ds$, and a target database $\Dt$, we say that \emph{$V$ is coherent
  with $\Ds$ and $\Dt$} if for each element $\aa$ in the source
alphabet, the extension of this element in $\Ds$ is contained in the
result of evaluating $V(\aa)$ over the database $\Dt$ (where $V(\aa)$
is the query that $V$ associates to $\aa$). Formally, $V$ is coherent
with $\Ds$ and $\Dt$ if for each $\aa\in\As$, $\aa^{\Ds}\subseteq
V(\aa)^{\Dt}$.

Following~\cite{CDLV07}, we say that a query $q_s\in\Qs$ is a \emph{sound
 rewriting}, or simply a \emph{rewriting}, of a query $q_t\in\Qt$ wrt views
$V$, if for every source database $\Ds$ and for every target database $\Dt$
such that $V$ is coherent with $\Ds$ and $\Dt$, we have that
$q_s^{\Ds}\subseteq q_t^{\Dt}$.  If $q_s^{\Ds}=q_t^{\Dt}$, the rewriting is
said to be \emph{exact}.  Further, we say that $q_s$ is \emph{empty wrt $V$} if
for every source database $\Ds$ and for every target database $\Dt$ such that
$V$ is coherent with $\Ds$ and $\Dt$, we have that $q_s^{\Ds}=\emptyset$.
Notice that, if all views in $V$ are empty (i.e., for each
$\aa\in\As$, $V(\aa)$ is the empty query), then trivially $q_s$ is
empty wrt $V$.  However $q_s$ may be empty wrt $V$ even in
the case in which all (or some) views are non-empty.

%% Following \cite{CDLV07}, we say that a query $q_s\in\Qs$ is a
%% \emph{sound rewriting}, or simply a \emph{rewriting}, of a query
%% $q_t\in\Qt$ wrt views $V$, if for every source database $\Ds$ and for
%% every target database $\Dt$ such that $\aa^{\Ds}\subseteq
%% V(\aa)^{\Dt}$ for each $\aa\in\As$, we have that $q_s^{\Ds}\subseteq
%% q_t^{\Dt}$.  We say that $q_s$ is \emph{empty wrt $V$} if for every
%% source database $\Ds$ and for every target database $\Dt$ such that
%% $\aa^{\Ds}\subseteq V(\aa)^{\Dt}$, for each $\aa\in\As$, we have that
%% $q_s^{\Ds}=\emptyset$.
%%
%% Notice that, if all views in $V$ are empty (i.e., the empty query), then
%% trivially $q_s$ is empty wrt $V$.  However $q_s$ may be empty wrt $V$ even in
%% the case in which all (or some) views are non-empty.
%%
%% Further, we say that $q_s$ is an \emph{exact rewriting of $q_t$ wrt $V$} if
%% $q_s^{\Ds}=q_t^{\Dt}$.
%%

We observe that, when $\Qs$ and $\Qv$ are monotonic query languages, the above
definitions of sound and exact rewritings are equivalent to the ones where the
notion of ``$V$ being coherent with $\Ds$ and $\Dt$'' is replaced by the
condition: for each $\aa\in\As$, $\aa^{\Ds}=V(\aa)^{\Dt}$
(see~\cite{CDLV07}). It is easy to see that, under this monotonic assumption,
$q_s$ is a rewriting of $q_t$ wrt views $V$ if $q_s[V]\sqsubseteq q_t$, where
here and in the following we use $q_s[V]$ to denote the query over $\At$
obtained from $q_s$ by substituting each source symbol $\aa\in\As$ with the
query $V(\aa)$.  Further, $q_s$ is empty wrt $V$ if $q_s[V]\equiv\emptyset$,
and it is an exact rewriting wrt $V$ if $q_s[V]\equiv q_t$. Note that in all
the settings considered in the next sections, the languages $\Qs$ and $\Qv$ are
monotonic.

%% condition $\aa^{\Ds}\subseteq V(\aa)^{\Dt}$ is replaced by
%% $\aa^{\Ds}=V(\aa)^{\Dt}$: This in turn reduces to saying that $q_s$ is a
%% rewriting of $q_t$ wrt views $V$ if $q_s[V]\sqsubseteq q_t$, where here and
%% in the following we use $q_s[V]$ to denote the query over $\At$ obtained
%% from $q_s$ by substituting each source symbol $\aa\in\As$ with $V(\aa)$.
%% Further, $q_s$ is an empty rewriting if $q_s[V]\equiv\emptyset$, and an
%% exact rewriting if $q_s[V]\equiv q_t$.

We are now ready to come back to the notion of ``views capturing a set of
mappings''.  We say that views $V$ \emph{capture} mappings $M$ if for each
$\map{q_s}{q_t}\in M$, the query $q_s$ is a rewriting of $q_t$ wrt $V$ and is
non-empty wrt $V$.  Analogously, we say that views $V$ \emph{exactly capture}
$M$ if for each mapping $\map{q_s}{q_t}\in M$, the query $q_s$ is an exact
rewriting of $q_t$ wrt $V$ and is non-empty wrt $V$.

We are now ready to introduce the (exact) view-synthesis and the
(exact) view-existence problems formally.

\begin{definition} The \emph{(exact) view-synthesis} problem is defined as
  follows: given a set $M$ of mappings, find views $V$ (exactly)
  capturing $M$.

  The \emph{(exact) view-existence} problem is defined as follows:
  given a set $M$ of mappings, decide whether there exist views $V$
  (exactly) capturing $M$.
\end{definition}

%% We will also consider the variant, called the \emph{exact view-synthesis}
%% problem, in which we further require that $q_s$ is an exact rewriting of
%% $q_t$ wrt $V$.  In this case, we say that views $V$ \emph{exactly capture}
%% $M$.
%%
%% The corresponding decision problem, i.e., given a set $M$ of mappings decide
%% whether there exist views $V$ (exactly) capturing $M$, is called the
%% \emph{(exact) view-existence} problem.

Finally, we also consider \emph{maximal views} capturing mappings $M$,
which are views $V$ such that there is no view $V'$ capturing $M$ such
that $(i)$ $V(a)\sqsubseteq V'(a)$ for every $a\in\As$, and $(ii)$
$V(a)\not\equiv V'(a)$ for some $a\in\As$.

%%% Local Variables:
%%% mode: latex
%%% TeX-master: "main"
%%% save-place: t
%%% End:

\section{View Synthesis for (U)CQ\lowercase{s}}
\label{sec:UCQs}

We start our investigations by tackling the case of view-synthesis and
view-existence for conjunctive queries (CQs) and their unions (UCQs).

We start with the case where $\Qs$, $\Qt$, and $\Qv$ are CQs, and establish the
following upper bounds.

\begin{theorem}\label{thm:CQ-CQ-CQ}
  In the case where $\Qs$, $\Qt$, and $\Qv$ are CQs, the view-existence and the
  exact view-existence problems are in NP.
\end{theorem}

\begin{proof}
Consider a mapping $q_s\leadsto q_t$, where $q_t$ contains $\ell$ atoms, and
views $V$ such that $q_s[V]\sqsubseteq q_t$.  Then, there exists a containment
mapping from $q_t$ to $q_s[V]$, and at most $\ell$ atoms of $q_s[V]$ will be in
the image of this containment mapping.  Hence, for each symbol $\aa\in\As$
occurring in $q_s$, only at most $\ell$ atoms in query $V(\aa)$ are needed to
satisfy the containment mapping.  In general, for a set $M$ of mappings, in
order to satisfy all containment mappings from $q_t$ to $q_s[V]$, for each
$q_s\leadsto q_t\in M$, we need in the query $V(\aa)$ at most
$\ell_M=\sum_{q_s\leadsto q_t\in M} \ell_{q_t}$ atoms, where $\ell_{q_t}$ is
the number of atoms in $q_t$.
Hence, in order to synthesize the views $V$, it suffices to guess, for each
symbol $\aa\in\As$ appearing in one of the mappings in $M$, a CQ $V(\aa)$
over $\At$ of size at most $\ell_M$, and check that $q_s[V]\sqsubseteq q_t$ (and
$q_t\sqsubseteq q_s[V]$ for the exact variant), for each $q_s\leadsto q_t\in M$.
This gives us immediately an NP upper bound for the (exact) view-existence
problem.
\end{proof}

In the case where $\Qs$ and $\Qt$ are UCQs, we can generalize the above
argument by considering containment between UCQs instead of containment between
CQs.

\begin{theorem}\label{thm:UCQ-UCQ-CQ}
  In the case where $\Qs$ and $\Qt$ are UCQs and $\Qv$ is CQs, the
  view-existence and the exact view-existence problems are in NP.
\end{theorem}

\begin{proof}
Consider a mapping $q_s\leadsto q_t$ and
views $V$ such that $q_s[V]\sqsubseteq q_t$.
We have that $q_s[V]\sqsubseteq q_t$ if for each CQ $q_1$ in the UCQ
$q_s[V]$ there is a CQ $q_2$ in the UCQ $q_t$ such that $q_1\sqsubseteq q_2$.
For a set $M$ of mappings, in
order to satisfy all containment mappings from $q_t$ to $q_s[V]$, for each
$q_s\leadsto q_t\in M$, we need in the query $V(\aa)$ at most
$\ell_M=\sum_{q_s\leadsto q_t\in M} \ell_{q_t}$ atoms, where  $\ell_{q_t}$ (this time)
is the maximum number of atoms in each of the CQs in $q_t$.
Hence the upper bound on the number of atoms of $V(\aa)$ is
$\ell_M=\sum_{q_s\leadsto q_t\in M} \ell_{q_t}$.
Again, In order to synthesize the views $V$, it suffices to guess, for each
symbol $\aa\in\As$ appearing in one of the mappings in $M$, a CQ $V(\aa)$ over
$\At$ of size at most $\ell_M$, and check that $q_s[V]\sqsubseteq q_t$ (and
$q_t\sqsubseteq q_s[V]$ for the exact variant), for each $q_s\leadsto q_t\in
M$.
\end{proof}

The last case we consider is the one where, in addition to $\Qs$ and $\Qt$,
also $\Qv$ is UCQs.  As for view-existence, we observe that the problem admits
a solutions for UCQs views iff it admits a solution for CQs views.

\begin{lemma}\label{thm:UCQ-UCQ-UCQ-CQ}
  An instance of the view-existence problem admits a solution in the case where
  $\Qs$, $\Qt$ and $\Qv$ are UCQs and $\Qv$ iff it admits solution in the case
  where $\Qs$ and $\Qt$ are UCQs and $\Qv$ is CQs.
\end{lemma}

\begin{proof}
Indeed, let $V$ be a set of UCQ views such that $q_s[V]\sqsubseteq q_t$ for
each mapping $q_s\leadsto q_t\in M$.  For such a mapping, $q_s[V]$ is a
nonempty positive query.
%% obtained from $q_s$ by substituting each atom $\alpha$ in $q_s$ with the UCQ
%% $\alpha[V]$.
Consider the views $V'$ obtained from $V$ by choosing, for each symbol $\aa$
in $\As$, as $V'(\aa)$ one of the CQs in $V(\aa)$.  Then, each nonempty
CQ in $q_s[V']$ is contained in $q_s[V]$, and hence in $q_t$.  It follows that
also views $V'$ provide a solution to the view-synthesis problem.
\end{proof}

Hence by the above lemma, we trivially get:
\begin{theorem}\label{thm:UCQ-UCQ-UCQ-existence}
  In the case where $\Qs$, $\Qt$, and $\Qv$ are UCQs, the view-existence
  problem is in NP.
\end{theorem}

As for exact view-existence, allowing for views that are UCQs changes indeed
the problem.

\begin{theorem}\label{thm:UCQ-UCQ-UCQ-exact}
  In the case where $\Qs$, $\Qt$, and $\Qv$ are UCQs, the exact view-existence
  problem is in $\Pi^p_2$.
\end{theorem}

\begin{proof}
Let $V$ be a set of UCQ views such that $q_s[V]=q_t$ for each
mapping $q_s\leadsto q_t\in M$.  Let us first consider one such mapping
$q_s\leadsto q_t$, and let $m_{q_t}$ be the number of CQs in $q_t$, and
$\ell_{q_t}$ the maximum number of atoms in each of the CQs in $q_t$.  Since
$q_s[V]\sqsubseteq q_t$, there is a containment mapping from each of the
$m_{q_t}$ CQs in $q_t$ to some CQ in the UCQ $q'_{s,V}$, where $q'_{s,V}$ is
obtained from $q_s[V]$ by distributing, for each atom $\alpha$ of $q_s$, the
unions in the UCQ $\alpha[V]$ over the conjunctions of each CQ of $q_s$.  Hence,
for each symbol $\aa\in\As$ occurring in $q_s$, only at most $m_{q_t}$ CQs
of at most $\ell_{q_t}$ atoms in query $V(\aa)$ are needed to satisfy the
containment mappings.  It follows that, to check the existence of UCQs views
$V$ and of such a containment mapping, it suffices to guess for each $\aa$ a
UCQ over $\At$ consisting of at most $m_{q_t}$ CQs, each with at most
$\ell_{q_t}$ atoms.
When considering all mappings $q_s\leadsto q_t\in M$, similar to the case
above, we have to use instead of $m_{q_t}$ and $\ell_{q_t}$, the sum of these
parameters over all mappings in $M$.
To check whether these views satisfy $q_t\sqsubseteq q_s[V] $, it suffices to
check for the existence of a containment mapping from $q_s[V]$ to each of the
CQs in $q_t$, which can be done in NP in the size of $q_t$.
To check whether these views satisfy $q_s[V]\sqsubseteq q_t$, we have to check
whether for each CQ $q'$ obtained by selecting one of the CQs $q''$ in $q_s$
and then substituting each atom $\alpha$ in $q''$ with one of the CQs in
$\alpha[V]$, there is a containment mapping from some CQ in $q_t$ to $q'$.
We can do so by a coNP computation that makes use of an NP oracle to check for
existence of a containment mapping.  This gives us the $\Pi^p_2$ upper bound.
\end{proof}

%%% Local Variables:
%%% mode: latex
%%% TeX-master: "main"
%%% save-place: t
%%% End:

\section{Tree-based Solution for RPQ\lowercase{s}}
\label{sec:tree-solution}

We address now the view-synthesis problem when $\Qs$, $\Qt$, and
$\Qv$ are RPQs, and present a techniques based on
tree automata on infinite trees~\cite{VaWo86}.  Specifically, we
consider automata running over complete labeled $\Sigma$-trees (i.e.,
trees in which the set of nodes is the set of all strings in $\Sigma^*$).

First, we observe that every language $L$ over an alphabet $\Sigma$ can
be represented as a function $L:\Sigma^*\ra\{0,1\}$, which, in turn, can be
considered as a $\{0,1\}$-labeling of the complete $\Sigma$-tree.
Consider a source alphabet $\As=\{\aa_1,\dots,\aa_n\}$ and the target
alphabet $\At$.  We can represent the views defined on $\As$ by the
$\{0,1\}^n$-labeled $\At$-tree $T_V$ (i.e., a $\At$-tree in which each
node is labeled with an $n$-tuple of elements of $\{0,1\}$) in which the
nodes representing the words in $V(\aa_i)$ are exactly those whose label
has~1 in the $i$-th component. We call such trees \emph{view trees}. Note
that views defined by view trees assign an arbitrary languages on
$\At$ to each source relation; these languages need not, a priori, be
regular.  We return to this point later.

Given a mapping $m=\map{q_s}{q_t}$, we construct now a tree automaton
$A_m$ accepting all view trees representing views $V$ capturing~$m$.
%%% DIEGO changed: 13/12/2009
%% (Checking nonemptiness is easy; there has to be at least one
%% node in the tree whose label is not $0^n$.)
Concerning the check that $q_s$ is not empty wrt $V$, we observe that, if there
is a word $w=c_1\cdots c_k$ in $\L(A_s)$ such that for all the letters
$a_{i_1},\ldots,a_{i_l}$ appearing in $w$, there are nodes in the tree where
the $i_j$'s component of the label is~1. The tree autmaton has to guess a set
of letters in $\As$ that cover a word accepted by $A_s$ (we can ignore the
letters in $\At$), and then check the above condition.
%%% END changed

We assume that $q_s$ is represented as an NWA
$A_s=(S_s,\As\cup\At,\pp^0_s,\delta_s,F_s)$ and $q_t$ is
represented as an NWA $A_t=(S_t,\At,\pp^0_t,\delta_t,F_t)$.%
\footnote{Transition functions of NWAs can be extended to sets of
states and words in a standard way.}
An \emph{annotation} for a view tree $T_V$ is a ternary relation
$\alpha \subseteq S_t^2 \times\As$. An annotation $\alpha$ is
\emph{correct} for $T_V$ if the following holds:
$(p,p',\aa_i)\in\alpha$ iff there is a word $w\in\At^*$ such
that $T_V(w)[i]=1$ and $p'\in\delta_t(p,w)$. Intuitively, $\alpha$
describes the transitions that $T_V$ can induce on $A_t$.

We say that an annotation $\alpha$ captures $\map{q_s}{q_t}$ if
for every word $w=c_1\cdots c_k$ in $\L(A_s)$ there is a sequence
$p_0,\ldots,p_{k+1}$ of states of $A_t$ such that $p_0=\pp^0_t$,
$p_{k+1}\in F_t$, and, for $i\in\{0,\ldots,k\}$, if $c_i\in\At$ then
$p_{i+1}\in\delta_t(p_i,c_i)$, and if $c_i=\aa_j\in\As$, then
$(p_i,p_{i+1},\aa_j)\in\alpha$.

The significance of an annotation capturing a mapping comes from the
following lemma.
\begin{lemma}\label{captures1}
  $V$ captures $\map{q_s}{q_t}$ iff there is an annotation $\alpha$ that is
  correct for $T_V$ and captures $\map{q_s}{q_t}$.
\end{lemma}

We now characterize when $\alpha$ captures $\map{q_s}{q_t}$.
\begin{lemma}\label{captures2}
$\alpha$ does not capture $\map{q_s}{q_t}$ iff there is a word $w=c_1
\cdots c_k$ in $\L(A_s)$ and a sequence $P_0,\ldots,P_{k+1}$ of sets
of states of $A_t$, such that $P_0=\{\pp^0_t\}$,
$P_{k+1}\cap F_t=\emptyset$, and for $i\in\{0,\ldots,k\}$, if $c_i\in\At$
then $P_{i+1}=\delta_t(P_1,w_i)$, and if $c_i=\aa_j\in\As$,
$p\in P_i$, and $(p,p',\aa_j)\in\alpha$, then $p'\in P_{i+1}$.
\end{lemma}
Thus, checking that $\alpha$ does not capture $\map{q_s}{q_t}$ can be
done by guessing the word $w$ and the sequence $P_0,\ldots,P_{k+1}$ of sets of
states of $A_t$ and checking the conditions.
This can be done in space logarithmic in
$A_s$ and polynomial in $A_t$. It follows that
we can check that an annotation $\alpha$ captures
$\map{q_s}{q_t}$ in time that is polynomial in $A_s$ and
exponential in~$A_t$.

We now describe a tree automaton $A_m$ that accepts precisely
the view trees $T_V$, where $V$ captures $m=\map{q_s}{q_t}$. By
Lemma~\ref{captures1}, all $A_m$ has to do is guess an annotation
$\alpha$ that captures $m$ and check that it is correct for $T_V$.

\begin{lemma}\label{captures3}
Given $A_s$ and $A_t$, we can construct a tree automaton $A_m$ that
accepts all view trees that capture $m=\map{q_s}{q_t}$. The size of
$A_m$ is exponential in the sizes of $A_s$ and $A_t$.
\end{lemma}
\begin{proof}
We construct $A_m=(S_m,\Sigma_m,\pp^0_m,\delta_m,F_m)$ as a B\"uchi
automaton on infinite trees \cite{VaWo86}. Recall that
$\Sigma_m=\{0,1\}^n$.  The state set is
$S_m=(2^{S_t^2\times\As})^2\times 2^{S_t^2} $.
That is, each state is a triple consisting of a pair of annotations
and a binary relation on $S_t$.
The initial state set $S^0_m$ consists of all triples
$\beta=(\alpha,\alpha,R_=)$, where $\alpha$ captures~$m$ and
$R_=\{(p,p,\aa)\mid p\in S_t\}$. Intuitively, an initial state
is a guess of an annotation. The automaton $A_m$ now has to check its
correctness; the second and third component of the state are used
for ``bookkeeping.''

Let $\At=\{b_1,\ldots,b_k\}$.
Then $(\beta_1,\ldots,\beta_k)\in\delta_m(\beta,c)$,
where $c=(c_1,\ldots,c_n)$, $\beta=(\alpha^1,\alpha^2,\alpha^3)$,
and $\beta_j=(\alpha_j^1,\alpha_j^2,\alpha_j^3)$ for $j\in\{1,\ldots,k\}$,
if the following hold:
\begin{enumerate}
\item
If $(p_1,p_2)\in\alpha^3$ and $c_i=1$, then $(p_1,p_2,\aa_i)\in\alpha^1$.
\item
$\alpha^1_j=\alpha^1$; that is, the first component does not change.
\item
$\alpha^3_j=\{(p_1,p_2')\mid (p_1,p_2)\in\alpha^3\mbox{ and }
p_2'\in\delta_t(p_2,b_j)\}$; that is, the third component remembers
paths between states of $A_t$.
\item
If $(p_1,p_2,\aa_i)\in\alpha^2$, then either $p_1=p_2$ and $c_i=1$, or,
for some $j\in\{1,\ldots,m\}$ and $p_1'\in\delta_t(p_1,b_j)$, we have
that $(p_1',p_2,\aa_i)\in\alpha^2_j$.
\end{enumerate}
Thus, the second component of the state helps to check
that all the paths in $A_t$ predicted by the guessed annotation are
fulfilled in the tree, while the third component helps to check that
all the paths that do occur in the tree are predicted by the guessed
annotation.  This means that the second component must ultimately
become empty.  Note that once it becomes empty, it can stay empty.
Thus the set $F_m$ of accepting states consists of all triples
of the form $(\alpha,\emptyset,R)$.

Note that the number of states of $A_m$ is exponential in the number of states
of $A_t$ and exponential in the alphabet of $A_s$. The alphabet of $A_m$ is
exponential in the size of the alphabet of $A_s$.
\end{proof}

\begin{theorem}\label{thm:RPQ-tree}
  In the case where $\Qs$, $\Qt$, and $\Qv$ are RPQs, the view existence
  problem is \EXPTIME.
\end{theorem}

\begin{proof}
We showed how to construct, with an exponential blowup a B\"uchi tree
automaton that accept all view trees that capture $m=\map{q_s}{q_t}$.
Note that computing the set of initial states, requires applying
Lemma~\ref{captures2}, and takes exponential time.
To handle a set $M$ of mappings, we simply take the product of these
automata (see product construction in~\cite{VaWo86}). To check that
the views are nonempty, we take the product with a very simple
automaton that checks that one of the labels of the tree is not
identically 0. We thus obtain a B\"uchi tree automaton $A_M$ that accepts
all view trees that represents nonempty views that capture $M$.

We can now check the nonemptiness of $A_M$ in quadratic time~\cite{VaWo86}.
If $\L(A_M)=\emptyset$, then the answer to the
view-existence problem is negative. If $\L(A_M)\neq\emptyset$, then
the nonemptiness algorithm returns a witness in the form of a
\emph{transducer} $A=(S,\At,\Sigma_m,\pp_0,\delta,\gamma)$,
where $S$ is a set of states (which is a subset of the state set of
the tree automaton), $\At$ is the input alphabet,
$\Sigma_m=\{0,1\}^n$ is the output alphabet, $\pp_0$ is a start state,
$\delta: S\times\At\rightarrow S$ is a deterministic transition
function, and $\gamma: S \rightarrow \Sigma_m$ is the output
function. From this transducer we can obtain an RPQ for each letter
$\aa_i\in\As$, represented by the DWA $A=(S,\At,\pp_0,\delta,F_i)$,
where $F_i=\{p\mid p\in S \mbox{ and } \gamma(p)[i]=1\}$.
\end{proof}

Note that the proof of Theorem~\ref{thm:RPQ-tree} implies that, wrt the
view-existence problem, considering views that are RPQs (as opposed to general,
possibly non-regular, path languages) is not a restriction, since the existence
of general views implies the existence of regular ones.  In fact, a similar
result holds also for the exact view-existence problem, as follows from the
results in the next section.  This is also in line with a similar observation
holding for the existence of rewritings of RPQs wrt RPQ views~\cite{CDLV02c}.

A final comment regarding maximal views. A view tree $T_V$ is maximal with
respect to a set $M$ of mappings if $V$ captures $M$, but flipping in one of
the labels a single $0$ to $1$ would destroy that property. Our tree-automata
techniques can be extended to produce maximal views, by quantifying over all
such flippings. This, however, would imply an additional exponential increase
in the complexity of the algorithm.

%%% Local Variables:
%%% mode: latex
%%% TeX-master: "main"
%%% save-place: t
%%% End:

\section{Congruence Class Based Solution for RPQ\lowercase{s}}
\label{sec:cc-solution}

We present now an alternative technique for view-synthesis for RPQs that will
allow us also to extend our results to more expressive forms of queries.  Our
solution is based on the characterization of regular languages by means of
congruence classes~\cite{Pin97}.

We start by showing that we can reduce the (exact) view-synthesis problem with
a set of mappings $M$ to an (exact) view-synthesis problem with a single
mapping $m$.

\begin{theorem}
  \label{thm:rpqs-single-mapping}
  Given a set $M$ of RPQ mappings, there is a single RPQ mapping $m$ such that,
  for every set $V$ of RPQ views, $V$ (exactly) captures $M$ iff $V$ (exactly)
  captures $m$.
\end{theorem}

\begin{proof}
%% Indeed, given a source alphabet $\As$, a target alphabet $\At$,
Let $M=\{\map{q_{0,s}}{q_{0,t}},\ldots,\map{q_{h,s}}{q_{h,t}}\}$ be the set of
mappings from $\As\cup\At$ to $\At$, and let $\At'=\At\cup\{\#\}$, where $\#$
is a fresh target symbol not occurring in $\As$ and $\At$.  We define a mapping
$m=\map{q_{M,s}}{q_{M,t}}$ from $\As\cup\At'$ to $\At'$, by setting
$q_{M,s}=q_{0,s}{\cdot}\#{\cdot}q_{1,s}{\cdot}\#\cdots\#{\cdot}q_{h,s}$ and
$q_{M,t}=q_{0,t}{\cdot}\#{\cdot}q_{1,t}{\cdot}\#\cdots\#{\cdot}q_{h,t}$.
Intuitively, the fresh symbol $\#$ acts as a separator between the different
parts of $q_{M,s}$ and $q_{M,t}$.  It is easy to verify that
$q_{i,s}[V]\sqsubseteq q_{i,t}$, for $i\in\{1,\ldots,h\}$ iff
$q_{M,s}[V]\sqsubseteq q_{M,t}$.
\end{proof}

Hence, w.l.o.g., in the following we will consider only the case where there is
a single mapping $\map{q_s}{q_t}$.

Let $A_t=(S_t,\At,\pp^0_t,\delta_t,F_t)$ be an NWA for $q_t$.  Then $A_t$
defines a set of (left-right) congruence classes partitioning $\At^*$.  Note
that the standard treatment of congruence classes is done with deterministic
automata~\cite{Pin97}, but we do it here with NWAs to avoid an exponential
blow-up.  For a word $w\in\At^*$, we denote with $[w]_{A_t}$ the congruence
class to which $w$ belongs.  Each congruence class is characterized by a binary
relation $R\subseteq S_t\times S_t$, where the congruence class associated with
$R$ is $C_R=\{w\in\At^*\mid \pp_2\in\delta_t(\pp_1,w,) \mbox{ iff }
(\pp_1,\pp_2)\in R\}$.  Intuitively, each word $w\in C_R$ connects
$\pp_1$ to $\pp_2$ in $A_t$, for each pair $(\pp_1,\pp_2)\in R$.

It follows immediately from the characterization of the congruence classes in
terms of binary relations over the states of $A_t$ that the set of congruence
classes is closed under concatenation.  Specifically, for two congruence
classes $C_{R_1}$ and $C_{R_2}$, respectively with associated relations $R_1$
and $R_2$, the binary relation associated with $C_{R_1}\cdot C_{R_2}$ is
$R_1\circ R_2$.\footnote{We use $L_1\cdot L_2$ to denote concatenation between
 languages, and $R_1\circ R_2$ to denote composition of binary relations.}
As a consequence, the set $\R$ of binary relations associated with the
congruence classes is $\R=2^{S_t\times S_t}$.  Let
$R_{\varepsilon}=\{(\pp,\pp)\mid \pp\in S_t\}$ and
$R_\bb=\{(\pp_1,\pp_2)\mid \pp_2\in\delta_t(\pp_1,\bb)\}$, for
each $\bb\in\At$.
%%
%% We set:
%% \begin{tabbing}
%%   \qquad\= $\R ~\leftarrow~ \{\R_{\varepsilon}\}$;\\
%%   \> repeat\\
%%   \> \quad \= $\R_\mathit{old} ~\leftarrow~ \R$;\\
%%   \> \> $\R ~\leftarrow~
%%          \R \cup \{R\circ R_\bb \mid R\in\R \mbox{ and } \bb\in\At\}$\\
%%   \> until $\R = \R_\mathit{old}$.
%% \end{tabbing}
%%
Then, for each $R\in\R$, the congruence class $C_R$ associated with $R$ is
accepted by the deterministic word automaton
$A_R=(\R,\At,R_{\varepsilon},\delta_{\sim},\{R\})$, where
$\delta_{\sim}(R,\bb)=R\circ R_\bb$, for each $R\in\R$ and $\bb\in\At$.  Notice
that, if $A_t$ has $m$ states, then the number of states of $A_R$ is $2^{m^2}$.

Let us consider the (non-exact) view-synthesis problem.
We observe first that we need to allow for the presence of empty queries for
the views.  Consider, e.g., $q_s=(a_1+a_3)\cdot(a_2+a_3)$ and $q_t=b_1\cdot
b_2$.  It is easy to see that the only views capturing $\map{q_s}{q_t}$ are
\[
  V(a_1) = b_1, \qquad
  V(a_2) = b_2, \qquad
  V(a_3) = \emptyset.
\]
Observe also that $b_1=[b_1]_{A_t}$ and  $b_2=[b_2]_{A_t}$, where $A_t$ is the
obvious NWA for $b_1\cdot b_2$.

We now prove two lemmas that will be used in the following.  The first lemma
states that w.l.o.g.\ we can restrict the attention to views capturing the
mapping that are \emph{singleton views}, i.e., views that are either empty or
constituted by a single word.

\begin{lemma}\label{lem:ssd-view-existence-word}
  Let $q_s$ be an RPQ over $\As\cup\At$, and $q_t$ an RPQ over $\At$.  If there
  exist RPQ views $V$ capturing $\map{q_s}{q_t}$, then there exist views $V'$
  capturing $\map{q_s}{q_t}$ such that each view in $V'$ is either a single
  word over $\At$ or empty.
\end{lemma}

\begin{proof}
Since $q_s[V]\not\equiv\emptyset$ and $q_s[V]\sqsubseteq q_t$, there exists a
word $a_1\cdots a_k\in\L(q_s)$ and a word $w_1\cdots w_k\in\L(q_s[V])$ and
hence in $\L(A_t)$, where $w_j=\L(V(a_j))$.  To define new views $V'$, we
consider for each $a\in\As$ appearing in $a_1\cdots a_k$ a word $w^a\in V(a)$
and set $V'(a)=w^a$. Instead, for each $a\in\As$ not appearing in $a_1\cdots
a_k$, we set $V'(a)=\emptyset$.  Now, $q_s[V']$ is nonempty by construction,
and since $V'(a)\sqsubseteq V(a)$ for every $a\in\As$, we have that
$q_s[V']\sqsubseteq q_s[V]\sqsubseteq q_t$.
\end{proof}

The next lemma shows that one can close views under congruence.

\begin{lemma}\label{lem:ssd-rewriting-cc}
  Let $q_s$ be an RPQ over $\As\cup\At$, $q_t$ an RPQ over $\At$ expressed
  through an NWA $A_t$, and $V$ singleton views capturing $\map{q_s}{q_t}$.
  Then $V'$ defined such that
  \[
    \L(V'(a))=
    \begin{cases}
      [w^a]_{A_t}, & \text{if } V(a)=w^a\\
      \emptyset, & \text{if } V(a)=\emptyset.
    \end{cases}
  \]
  captures $\map{q_s}{q_t}$.
\end{lemma}

\begin{proof}
Consider a word $a_1\cdots a_h\in\L(q_s)$.  If there is one of the $a_i$ such
that $V(a_i)=\emptyset$, then
$\L(V(a_1))\cdots\L(V(a_h))=\emptyset\subseteq\L(q_t)$.  Otherwise, we have
that $\L(V(a_i))=\{w^{a_i}\}$, for $i\in\{1,\ldots,h\}$, and since
$w^{a_1}\cdots w^{a_h}\in\L(q_s[V])\subseteq\L(A_t)$, there is a
sequence $\pp_0,\pp_1,\ldots,\pp_h$ of states of $A_t$ such that
$\pp_0=\pp^0_t$, $\pp_h\in F_t$, and
$\pp_i\in\in\delta_t(\pp_{i-1},w^{a_i})$, for $i\in\{1,\ldots,h\}$.
Consider now, for each $i\in\{1,\ldots,h\}$, a word $w'_i\in
\L(V'(a_i))=[w^{a_i}]_{A_t}$.  Making use of the characterization of
$[w^{a_i}]_{A_t}$ in terms of a binary relation over $S_t$, we have for each
word in $[w^{a_i}]_{A_t}$, and in particular for $w'_i$, that
$\pp_i\in\delta_t(\pp_{i-1},w'_i)$.  Hence,
$\pp_h\in\delta_t(\pp_o^t,w'_1\cdots w'_h)$ and $w'_1\cdots w'_h\in\L(q_t)$.
\end{proof}

From these two lemmas we get that, when searching for views capturing the
mappings, we can restrict the attention to views that are congruence classes
for $A_t$.

\begin{lemma}\label{lem:ssd-view-existence-cc}
  Let $q_s$ be an RPQ over $\As\cup\At$, and $q_t$ an RPQ over $\At$ expressed
  through an NWA $A_t$.  If there exist RPQ views $V$ over $\At$ capturing
  $\map{q_s}{q_t}$, then there exist views $V'$ capturing $\map{q_s}{q_t}$ such
  that each view in $V'$ is a congruence class for $A_t$.
\end{lemma}

\begin{proof}
If there exist RPQ views $V$ over $\At$ capturing $\map{q_s}{q_t}$, then by
Lemma~\ref{lem:ssd-view-existence-word}, w.l.o.g., we can assume that $V$ are
singleton views.  Then, the claim follows from
Lemma~\ref{lem:ssd-rewriting-cc}.
\end{proof}

From the above lemma, we can immediately derive an \EXPTIME procedure for view
existence, which gives us an alternative proof of Theorem~\ref{thm:RPQ-tree}.
We first observe that, for an NWA $A_t$ with $m$ states, each view defined by a
congruence class $C_R$ for $A_t$ can be represented by the NWA $A_R$, which has
at most $2^{m^2}$ states.  For a set $V$ of views that are congruence classes,
we can test whether $q_s[V]\not \equiv\emptyset$ and $q_s[V]\sqsubseteq q_t$ by
\begin{itemize}
\item substituting each $a$-transition in the NWA $A_s$ for $q_s$ with the NWA
  $A_{R_a}$, where $C_{R_a}=\L(V(a))$, thus obtaining an NWA $A_{s,V}$;
\item complementing $A_t$, obtaining an NWA $\overline{A_t}$; and
\item checking the nonemptiness of $A_{s,V}$ and the emptiness of
  $A_{s,V}\times\overline{A_t}$.
\end{itemize}
Such a test can be done in time polynomial in the size of $A_s$ and exponential
in the size of $A_t$.

Considering that the number of distinct congruence classes is at most
$2^{m^2}$, the number of possible assignments of congruence classes to
$n$ view symbols occurring in $q_s$ is at most $2^{n\cdot m^2}$.  For
each such assignment defining views $V$, we need to test whether
$q_s[V]\not \equiv\emptyset$ and $q_s[V]\sqsubseteq q_t$.  Hence the
overall check for the view-existence problem requires time exponential
in the size of $A_t$, polynomial in the size of $A_s$ and exponential
in the number of source symbols occurring in $q_s$.

\medskip

The technique presented here based on congruence classes can be adapted to
address also the exact view existence problem.  The difference wrt to
(non-exact) view existence is that in this case we need to consider also views
that are unions of congruence classes.  Indeed, congruence classes (and hence
rewritings) are not closed under union, as shown by the following example.

Let $q_s=a_1\cdot a_2$ and $q_t=00+01+10$.  Then the following two sets of
incomparable views maximally capture $\map{q_s}{q_t}$:
\[
  \begin{array}[t]{l@{~}c@{~}l}
    V_1(a_1) &=& 0,\\
    V_1(a_2) &=& 0+1.
  \end{array}
  \qquad\qquad
  \begin{array}[t]{l@{~}c@{~}l}
    V_2(a_1) &=& 0+1,\\
    V_2(a_2) &=& 0.
  \end{array}
\]
Notice that views $V$, where $V(a_i)=V_1(a_i)+V_2(a_i)$, for $i\in\{1,2\}$,
does not capture $\map{q_s}{q_t}$, since $q_s[V]$ includes $11$.

On the other hand, we can show that considering views that are unions of
congruence classes is sufficient to obtain maximal unfoldings.
We first generalize Lemma~\ref{lem:ssd-rewriting-cc} to non-singleton views.

\begin{lemma}\label{lem:ssd-rewriting-union-of-cc}
  Let $q_s$ be an RPQ over $\As\cup\At$, $q_t$ an RPQ over $\At$ expressed
  through an NWA $A_t$, and $V$ a set of views capturing $\map{q_s}{q_t}$.
  Then $V'$ defined such that
  \[
    \L(V'(a))=
    \begin{cases}
      \bigcup_{w\in\L(V(a))}[w]_{A_t}, & \text{if } V(a)\neq\emptyset\\
      \emptyset, & \text{if } V(a)=\emptyset.
    \end{cases}
  \]
  captures $\map{q_s}{q_t}$.
\end{lemma}

\begin{proof}
Consider a word $a_1\cdots a_h\in\L(q_s)$.  If there is one of the $a_i$ such
that $V(a_i)=\emptyset$, then
$\L(V(a_1))\cdots\L(V(a_h))=\emptyset\subseteq\L(q_t)$.  Otherwise, we have
that, for $i\in\{1,\ldots,h\}$, for some $w^{a_i}\in\L(V(a_i))$, the word
$w^{a_1}\cdots w^{a_h}\in\L(q_s[V])\subseteq\L(A_t)$.  We show that, for each
$i\in\{1,\ldots,h\}$, we also have that $w^{a_1}\cdots w^{a_{i-1}}\cdot w'\cdot
w^{a_{i+1}}\cdots w^{a_h}\in\L(A_t)$, for each
$w'\in\bigcup_{w\in\L(V(a_i))}[w]_{A_t}$.  First, by definition of rewriting,
if $w^{a_1}\cdots w^{a_h}\in\L(q_s[V])\subseteq\L(A_t)$, then, for each
$w\in\L(V(a_i))$, we also have that $w^{a_1}\cdots w^{a_{i-1}}\cdot w\cdot
w^{a_{i+1}}\cdots w^{a_h}\in\L(q_s[V])\subseteq\L(A_t)$.
Then there is a sequence $\pp_0,\pp_1,\ldots,\pp_h$ of states of $A_t$
such that $\pp_0=\pp^0_t$, $\pp_h\in F_t$,
$\pp_j\in\delta_t(\pp_{j-1},w^{a_j})$, for
$j\in\{1,\ldots,i{-}1,i{+}1,\ldots,h\}$, and
$\pp_i\in\delta_t(\pp_{i-1},w)$.
Then, by the definition of congruence classes, for each word $w'\in[w]_{A_t}$,
we have that $\pp_i\in\delta_t(\pp_{i-1},w'_i)$, and hence
$w^{a_1}\cdots w^{a_{i-1}}\cdot w'\cdot w^{a_{i+1}}\cdots w^{a_h}\in\L(A_t)$.
\end{proof}

The above lemma implies that, when searching for views that maximally capture
the mappings, we can restrict the attention to views that are unions of
congruence classes.

\begin{lemma}\label{lem:rpq-union-congruence-classes}
  Given a mapping $m=\map{q_s}{q_t}$, where $q_t$ is defined by an NWA $A_t$,
  every set of views $V$ that maximally captures $m$ is such that each view in
  $V$ is a union of congruence classes for $A_t$.
\end{lemma}

\begin{proof}
Consider a set of views $V$ that maximally captures $m$, and assume that for
some $a\in\As$, $V(a)$ is not a union of congruence classes for $A_t$. Then
there is some word $w\in\L(V(a))$ and some word $w'\in [w]_{A_t}$ such that
$w'\notin\L(V(a))$.  By Lemma~\ref{lem:ssd-rewriting-union-of-cc}, the set of
views $V'$ with $\L(V'(a))=\L(V(a))\cup\{w'\}$ also captures $m$, thus
contradicting the maximality of $V$.
\end{proof}

We get the following upper bound for the exact view existence problem.

\begin{theorem}\label{thm:rpq-exact-view-existence}
  In the case where $\Qs$, $\Qt$, and $\Qv$ are RPQs, the exact view existence
  problem is in \EXPSPACE.
\end{theorem}
\begin{proof}
% Considering that the number of distinct union of congruence classes is at
% most $2^{2^{m^2}}$, the number of possible assignments of union of congruence
% classes to $n$ view symbols occurring in $q_s$ is at most $2^{n\cdot
%  2^{m^2}}$.
By Lemma~\ref{lem:rpq-union-congruence-classes}, we can nondeterministically
choose views $V$ that are unions of congruence classes and then test whether
$q_t\equiv q_s[v]$ (we assume that $q_t\not\equiv\emptyset$, otherwise the
problem trivializes).  To do so, we build an NWA $A_{s,V}$ accepting
$\L(q_s[V])$ as follows.  We start by observing that for each union $U$ of
congruence classes, we can build the automaton
$A_U=(\R,\pp_t,R_\epsilon,\delta_\sim,U)$ accepting the words in $U$, which
incidentally, is deterministic.  Hence, by substituting each $a$-transition in
the NWA $A_s$ for $q_s$ with the NWA $A_{U_a}$, where $V(a)=U_{a}$, we obtain
an NWA $A_{s,V}$.  Note that, even when $A_s$ is deterministic $A_{s,V}$ may be
nondeterministic.

To test $q_s[V]\sqsubseteq q_t$, we complement $A_t$, obtaining the NWA
$\overline{A_t}$, and check the NWA $A_{s,V}\times\overline{A_t}$ for
emptiness.  The size of $A_{s,V}\times\overline{A_t}$ is polynomial in the size
of $\A_s$ and exponential in the size of $\A_t$.  Checking for emptiness
can be done in exponential time, and considering the initial nondeterministic
guess, we get a \NEXPTIME upper bound.

To test $q_t\sqsubseteq q_s[V]$, we complement $A_{s,V}$, obtaining the NWA
$\overline{A_{s,V}}$, and check $A_t\cap\overline{A_{s,V}}$ for emptiness.
Since $\overline{A_{s,V}}$ is nondeterministic, complementation is
exponential. However, we observe that such a complementation can be done on the
fly in \EXPSPACE, while checking for emptiness and intersecting with $A_t$.  As
a consequence, considering the initial nondeteministic guess, exact view
existence can be decided in \NEXPSPACE, which is equivalent to \EXPSPACE.
\end{proof}

%%% Local Variables:
%%% mode: latex
%%% TeX-master: "main"
%%% save-place: t
%%% End:
\section{Extensions}
\label{sec:extensions}

In this section we sketch the extension of the results of the previous
section to more expressive classes of queries: 2RQPs, CRPQs, UCRPQs, and
UC2RPQs.

\subsection{2RPQs}
\label{subsec-2RPQs}

Consider now the view-synthesis problem for the case where $\Qs$,
$\Qt$, and $\Qv$ are 2RPQs, expressed by means of NWAs over
the alphabets $\Sigma^\pm$ and $\Delta^\pm$.

A key concept for 2RPQs is that of \emph{folding}. Let $u,v\in\Sigma^\pm$.
We say that $v$ \emph{folds} onto $u$, denoted $v\rightsquigarrow u$, if $v$
can be ``folded'' on $u$, e.g., $abb^-bc\rightsquigarrow abc$.  Formally, we
say that $v=v_1 \cdots v_m$ folds onto $u=u_1 \cdots u_n$ if there is a
sequence $i_0,\ldots,i_m$ of positive integers between $0$ and $|u|$ such that
\begin{itemize}
\item
$i_0=0$ and $i_m=n$, and
\item
for $j\in\{0,\ldots,m\}$,
either $i_{j+1}=i_j+1$ and $v_{j+1}=u_{i_j+1}$,
or $i_{j+1}=i_j-1$ and $v_{j+1}=u^-_{i_{j+1}}$.
\end{itemize}
Let $L$ be a language over $\Sigma^\pm$. We define
$\mathit{fold}(L)=\{u~:~v\rightsquigarrow u, v\in L\}$.

A language-theoretic characterization for containment of 2RPQs was
provided in \cite{CDLV03c}:
\begin{lemma}\label{lang1}
Let $q_1$ and $q_2$ be 2RPQs. Then
$q_1 \sqsubseteq q_2$ iff $\L(q_1) \subseteq \mathit{fold}(\L(q_2))$.
\end{lemma}
Furthermore, it is shown in \cite{CDLV03c} that if
$A$ is an $n$-state NWA over $\Sigma^\pm$,  then there is a
2NWA for $\mathit{fold}(\L(A))$ with $n\cdot (|\Sigma^\pm|+1)$ states.
(We use 2NWA to refer to two-way automata on words.)

In the view-existence problem , we are given queries $q_s$ and $q_t$, expressed
as NWAs $A_s$ and $A_t$, and we are asked whether there exist
%%% DIEGO changed: 13/12/2009
%% nonempty 2RPQ views $V$ such that $q_s[V]\sqsubseteq q_t$.
nonempty 2RPQ views $V$ such that $q_s[V]\sqsubseteq q_t$ and such that
$q_s[V]\not\equiv\emptyset$.
%%% END changed
We can use Lemma~\ref{lang1} for the tree-automata solution.  A labeled tree
$V:(\Sigma^\pm) \rightarrow \{0,1\}^n$ represents candidate views. To check
that $q_s[V]\not\sqsubseteq q_t$, we check that $\L(A_s[V])\not\subseteq
\mathit{fold}(\L(A_t))$. We now proceed as in Section~\ref{sec:tree-solution},
using the 2NWA for $\mathit{fold}(\L(A_t))$ instead of $A_t$. This
requires first converting the 2NWA to an NWA with an exponential
blow-up~\cite{SaSi78}, increasing the overall complexity to 2\EXPTIME.

We can also use Lemma~\ref{lang1} for the congruence-based solution.
Here also a simplistic approach would be to convert the 2NWA for
$\mathit{fold}(\L(A_t))$ into an NWA, with an exponential blow-up,
and proceed as in Section~\ref{sec:cc-solution}.
To avoid this exponential blowup, we need an exponential bound on the number
of congruence classes. For an NWA, we saw that each congruence class can be
defined in terms of a binary relation over its set of states.  It turns out
that for a 2NWA $A$, a congruence class can be defined in terms of \emph{four}
binary relations over the set $S_t$ of states of $A$:
\begin{enumerate}
\item $R_{lr}$: a pair $(\pp_1,\pp_2)\in R_{lr}$ means that there is a
  word $w$ that leads $A$ from $\pp_1$ to $\pp_2$, where $w$ is entered
  on the left and exited on the right.
\item $R_{rl}$: a pair $(\pp_1,\pp_2)\in R_{rl}$ means that there is a
  word $w$ that leads $A$ from $\pp_1$ to $\pp_2$, where $w$ is entered
  on the right and exited on the left.
\item $R_{ll}$: a pair $(\pp_1,\pp_2)\in R_{ll}$ means that there is a
  word $w$ that leads $A$ from $\pp_1$ to $\pp_2$, where $w$ is entered
  on the left and exited on the left.
\item $R_{rr}$: a pair $(\pp_1,\pp_2)\in R_{rr}$ means that there is a
  word $w$ that leads $A$ from $\pp_1$ to $\pp_2$, where $w$ is entered
  on the right and exited on the right.
\end{enumerate}
Thus, the number of congruence classes when $A$ has $m$ states is
$2^{4m^2}$ rather than $2^{m^2}$, which is still an exponential.
This enables us to adapt the technique of Section~\ref{sec:cc-solution}
with essentially the same complexity bounds.

\begin{theorem}\label{thm:2RPQs}
  In the case where $\Qs$, $\Qt$, and $\Qv$ are 2RPQs, the view-existence
  problem is \EXPTIME and the exact view-existence problem is in \EXPSPACE.
\end{theorem}

\subsection{CRPQs}
\label{subsec-CRPQs}

Consider now the view-synthesis problem for the case where
%%% DIEGO changed: 13/12/2009
%% $\Qs$, $\Qt$, and $\Qv$ are CRPQs,
$\Qs$ and $\Qt$ are CRPQs,
%%% END changed
where the constituent RPQs are expressed by means of NWAs.
Here the views have to be RPQs, rather than CRPQs, since CRPQs are not closed
under substitutions. Thus, we can still represent views in terms of a labeled
tree $V: \Sigma^* \rightarrow \{0,1\}^n$.  The crux of our approach is to
reduce containment of two CRPQs, $q_1$ and $q_2$ to containment of standard
languages. This was done in~\cite{CDLV00b}.

Let $q_h$, for $h=\{1,2\}$, be in the form
\[
  q_h(x_1,\ldots,x_n) \leftarrow
    \begin{array}[t]{@{}l}
      q_{h,1}(y_{h,1},y_{h,2}) \land\cdots\land{}\\
      q_{h,m_h}(y_{h,2m_h-1},y_{h,2m_h})
    \end{array}
\]
and let $\V_1$, $\V_2$ be the sets of variables of $q_1$
and $q_2$ respectively. It is shown in \cite{CDLV00b} that the
containment $q_1\sqsubseteq q_2$ can be reduced to the containment
$\L(A_1)\subseteq\L(A_2)$ of two word automata $\A_1$ and $\A_2$.
$\A_1$ is an NWA, whose size is exponential in $q_1$ and it accepts
certain words of the form
\[
  \$ d_1 w_1 d_2 \$ d_3 w_2 d_4 \$ \cdots \$ d_{2m_1-1} w_{m_1} d_{2m_1}\$
\]
where each $d_i$ is a subset of $\V_1$ and the words $w_i$ are over the alphabet
of $A_1$. Such words constitute a linear representation of certain
semistructured databases that are canonical for $q_1$ in some sense.
$A_2$ is a 2NWA, whose size is an exponential in the size of $q_2$, and it
accepts words of the above form if the there is an appropriate
mapping from $q_2$ to the database represented by these words.
The reduction of the containment $q_1\sqsubseteq q_2$ to
$\L(A_1)\subseteq\L(A_2)$ is shown in~\cite{CDLV00b}.

We can now adapt the tree-automata technique of
Section~\ref{sec:tree-solution}.  From $q_s$ and $q_t$ we can construct
word automata $A_s$ and $A_t$ as in~\cite{CDLV00b}.  We now ask
if we have nonempty RPQ views $V$ such that $\L(A_s[V])\subseteq\L(A_t)$.
This can be done as in Section~\ref{sec:tree-solution}, after converting
$A_t$ to an NWA.

The ability to reduce containment of CRPQs to containment of word
automata means that we can also apply the congruence-class technique of
Section~\ref{sec:cc-solution}. Suppose that we have nonempty RPQ views
$V$ such that $\L(A_s[V])\subseteq\L(A_t)$. Then we can again assume that
the views are closed with respect to the congruence classes
of $A_t$.  Thus, the techniques of Section~\ref{sec:cc-solution} can
be applied.

\begin{theorem}\label{thm:C2RPQ}
  In the case where $\Qs$, $\Qt$, and $\Qv$ are \mbox{CRPQs}, the
  view-existence problem is in \textsc{2}\EXPTIME, and the exact view-existence
  problem is in \textsc{2}\EXPSPACE.
\end{theorem}

\subsection{UC2RPQs}\label{subsec-UC2RPQs}

Here we allow both C2RPQs and unions. Since UC2RPQs are not closed
under substitutions, we consider here 2RPQ views. The extension to
handle unions is straightforward.  To handle C2RPQs, we need to combine
the techniques of Sections~\ref{subsec-2RPQs} and~\ref{subsec-CRPQs}.
The key idea is the reduction of query containment to containment of
word automata. The resulting upper bounds are identical to those we
obtained for \mbox{CRPQs}.

%%% Local Variables:
%%% mode: latex
%%% TeX-master: "main"
%%% save-place: t
%%% End:

\section{Conclusions}
\label{sec:conclusions}

In this paper we have addressed the issue of synthesizing a set of
views starting from a collection of mappings relating a source schema
to a target schema.

We have argued that the problem is relevant in several scenarios, especially
data warehousing, data integration and mashup, and data exchange. We have
provided a formalization of the problem based on query rewriting, and we have
presented techniques and complexity upper bounds for two cases, namely,
relational data, and graph-based semistructured data.  We concentrated on the
basic problems of view-existence, and we have shown that in both cases the
problem is decidable, with different complexity upper bounds depending on the
types of query languages used in the mappings and the views, and on the variant
(sound or exact rewriting) of the problem.

We plan to continue investigating the view-synthesis problem along different
directions. First, we aim at deriving lower complexity bounds for the
view-existence problem.  Secondly, we are interested in studying view-synthesis
for tree-based (e.g., XML) semistructured data.  Finally, while in this paper
we have based the notion of view-synthesis on query rewriting, it would be
interesting to explore a variant of this notion, based on query answering using
views.  In this variant, views $V$ capture a mapping of the form
$\map{q_s}{q_t}$ if, for each source database $\Ds$, the query $q_s$ computes
the certain answers to $q_t$ wrt $V$ and $D_s$ \cite{CDLV07}.

%%% Local Variables:
%%% mode: latex
%%% TeX-master: "main"
%%% save-place: t
%%% End:

%%\section{Acknowledgments}

%% This research has been partially supported by the EU funded
%% Projects INFOMIX (IST-2001-33570) and TONES, by MIUR - Fondo Speciale per
%% lo Sviluppo della Ricerca di Interesse Strategico - project ``Societ\`a
%% dell'Informazione'', subproject SP1 ``Reti Internet: Efficienza,
%% Integrazione e Sicurezza'', by MIUR - Fondo per gli Investimenti della
%% Ricerca di Base (FIRB) - project ``MAIS: Multichannel Adaptive Information
%% Systems'', by project HYPER, funded by IBM through a Shared University
%% Research (SUR) Award grant, by NSF grants CCR-9988322, CCR-0124077,
%% CCR-0311326, IIS-9908435, IIS-9978135, EIA-0086264, and ANI-0216467, by
%% US-Israel BSF grant 9800096, by Texas ATP grant 003604-0058-2003, and by a
%% grant from the Intel Corporation.}%

%%\clearpage
%%\pagenumbering{roman}

\bibliographystyle{abbrv}
\bibliography{main-bib}

\begin{thebibliography}{10}

\bibitem{Abit97}
S.~Abiteboul.
\newblock Querying semi-structured data.
\newblock In {\em Proc.\ of the 6th Int.\ Conf.\ on Database Theory (ICDT'97)},
  pages 1--18, 1997.

\bibitem{AbDu98}
S.~Abiteboul and O.~Duschka.
\newblock Complexity of answering queries using materialized views.
\newblock In {\em Proc.\ of the 17th ACM SIGACT SIGMOD SIGART Symp.\ on
  Principles of Database Systems (PODS'98)}, pages 254--265, 1998.

\bibitem{AbHV95}
S.~Abiteboul, R.~Hull, and V.~Vianu.
\newblock {\em Foundations of Databases}.
\newblock Addison Wesley Publ.\ Co., 1995.

\bibitem{ACGK*09}
F.~N. Afrati, R.~Chirkova, M.~Gergatsoulis, B.~Kimelfeld, V.~Pavlaki, and
  Y.~Sagiv.
\newblock On rewriting {XPath} queries using views.
\newblock In {\em Proc.\ of the 12th Int.\ Conf.\ on Extending Database
  Technology (EDBT~2009)}, pages 168--179, 2009.

\bibitem{ArFM10}
P.~C. Arocena, A.~Fuxman, and R.~J. Miller.
\newblock Composing local-as-view mappings: Closure and applications.
\newblock In {\em Proc.\ of the 13th Int.\ Conf.\ on Database Theory
  (ICDT~2010)}, 2010.
\newblock To appear.

\bibitem{Bune97}
P.~Buneman.
\newblock Semistructured data.
\newblock In {\em Proc.\ of the 16th ACM SIGACT SIGMOD SIGART Symp.\ on
  Principles of Database Systems (PODS'97)}, pages 117--121, 1997.

\bibitem{CDLV00b}
D.~Calvanese, G.~De~Giacomo, M.~Lenzerini, and M.~Y. Vardi.
\newblock Containment of conjunctive regular path queries with inverse.
\newblock In {\em Proc.\ of the 7th Int.\ Conf.\ on the Principles of Knowledge
  Representation and Reasoning (KR~2000)}, pages 176--185, 2000.

\bibitem{CDLV00c}
D.~Calvanese, G.~De~Giacomo, M.~Lenzerini, and M.~Y. Vardi.
\newblock Query processing using views for regular path queries with inverse.
\newblock In {\em Proc.\ of the 19th ACM SIGACT SIGMOD SIGART Symp.\ on
  Principles of Database Systems (PODS~2000)}, pages 58--66, 2000.

\bibitem{CDLV02c}
D.~Calvanese, G.~De~Giacomo, M.~Lenzerini, and M.~Y. Vardi.
\newblock Rewriting of regular expressions and regular path queries.
\newblock {\em J.\ of Computer and System Sciences}, 64(3):443--465, 2002.

\bibitem{CDLV03c}
D.~Calvanese, G.~De~Giacomo, M.~Lenzerini, and M.~Y. Vardi.
\newblock Reasoning on regular path queries.
\newblock {\em {SIGMOD} Record}, 32(4):83--92, 2003.

\bibitem{CDLV07}
D.~Calvanese, G.~De~Giacomo, M.~Lenzerini, and M.~Y. Vardi.
\newblock View-based query processing: {O}n the relationship between rewriting,
  answering and losslessness.
\newblock {\em Theoretical Computer Science}, 371(3):169--182, 2007.

\bibitem{CDLV09}
D.~Calvanese, G.~De~Giacomo, M.~Lenzerini, and M.~Y. Vardi.
\newblock An automata-theoretic approach to {R}egular {XPath}.
\newblock In {\em Proc.\ of the 12th Int.\ Symp.\ on Database Programming
  Languages (DBPL~2009)}, volume 5708 of {\em Lecture Notes in Computer
  Science}, pages 18--35. Springer, 2009.

\bibitem{ChCS09}
D.~Chen, R.~Chirkova, and F.~Sadri.
\newblock Query optimization using restructured views: Theory and experiments.
\newblock {\em Information Systems}, 34(3):353--370, 2009.

\bibitem{ChHS01}
R.~Chirkova, A.~Y. Halevy, and D.~Suciu.
\newblock A formal perspective on the view selection problem.
\newblock In {\em Proc.\ of the 27th Int.\ Conf.\ on Very Large Data Bases
  (VLDB~2001)}, pages 59--68, 2001.

\bibitem{DPGW10}
A.~Das~Sarma, A.~Parameswaran, H.~Garcia-Molina, and J.~Widom.
\newblock Synthesizing view definitions from data.
\newblock In {\em Proc.\ of the 13th Int.\ Conf.\ on Database Theory
  (ICDT~2010)}, 2010.
\newblock To appear.

\bibitem{DLDH04}
R.~Dhamankar, Y.~Lee, A.~Doan, A.~Y. Halevy, and P.~Domingos.
\newblock {iMAP}: {D}iscovering complex mappings between database schemas.
\newblock In {\em Proc.\ of the ACM SIGMOD Int.\ Conf.\ on Management of Data},
  pages 383--394, 2004.

\bibitem{DHPB09}
G.~Di~Lorenzo, H.~Hacid, H.-Y. Paik, and B.~Benatallah.
\newblock Data integration in mashups.
\newblock {\em {SIGMOD} Record}, 38(1):59--66, 2009.

\bibitem{GiYS07}
F.~Giunchiglia, M.~Yatskevich, and P.~Shvaiko.
\newblock Semantic matching: Algorithms and implementation.
\newblock {\em J.\ on Data Semantics}, 9:1--38, 2007.

\bibitem{GrMe99}
G.~Grahne and A.~O. Mendelzon.
\newblock Tableau techniques for querying information sources through global
  schemas.
\newblock In {\em Proc.\ of the 7th Int.\ Conf.\ on Database Theory (ICDT'99)},
  volume 1540 of {\em Lecture Notes in Computer Science}, pages 332--347.
  Springer, 1999.

\bibitem{GrTh03a}
G.~Grahne and A.~Thomo.
\newblock Algebraic rewritings for optimizing regular path queries.
\newblock {\em Theoretical Computer Science}, 296(3):453--471, 2003.

\bibitem{Hale00}
A.~Y. Halevy.
\newblock Theory of answering queries using views.
\newblock {\em {SIGMOD} Record}, 29(4):40--47, 2000.

\bibitem{Hale01}
A.~Y. Halevy.
\newblock Answering queries using views: {A} survey.
\newblock {\em Very Large Database J.}, 10(4):270--294, 2001.

\bibitem{Inmo96}
W.~H. Inmon.
\newblock {\em Building the Data Warehouse}.
\newblock John Wiley \& Sons, second edition, 1996.

\bibitem{Kola05}
P.~G. Kolaitis.
\newblock Schema mappings, data exchange, and metadata management.
\newblock In {\em Proc.\ of the 24rd ACM SIGACT SIGMOD SIGART Symp.\ on
  Principles of Database Systems (PODS~2005)}, pages 61--75, 2005.

\bibitem{Lenz02}
M.~Lenzerini.
\newblock Data integration: {A} theoretical perspective.
\newblock In {\em Proc.\ of the 21st ACM SIGACT SIGMOD SIGART Symp.\ on
  Principles of Database Systems (PODS~2002)}, pages 233--246, 2002.

\bibitem{LMSS95}
A.~Y. Levy, A.~O. Mendelzon, Y.~Sagiv, and D.~Srivastava.
\newblock Answering queries using views.
\newblock In {\em Proc.\ of the 14th ACM SIGACT SIGMOD SIGART Symp.\ on
  Principles of Database Systems (PODS'95)}, pages 95--104, 1995.

\bibitem{Pin97}
J.-E. Pin.
\newblock Syntactic semigroups.
\newblock In G.~Rozenberg and A.~Salomaa, editors, {\em Handbook of Formal
  Language Theory}, volume~1, chapter~10, pages 679--746. Springer, 1997.

\bibitem{RaBe01}
E.~Rahm and P.~A. Bernstein.
\newblock A survey of approaches to automatic schema matching.
\newblock {\em Very Large Database J.}, 10(4):334--350, 2001.

\bibitem{SaSi78}
W.~Sakoda and M.~Sipser.
\newblock Nondeterminism and the size of two way finite automata.
\newblock In {\em Proc.\ of the 10th ACM Symp.\ on Theory of Computing
  (STOC'78)}, pages 275--286, 1978.

\bibitem{SeGo08}
P.~Senellart and G.~Gottlob.
\newblock On the complexity of deriving schema mappings from database
  instances.
\newblock In {\em Proc.\ of the 27th ACM SIGACT SIGMOD SIGART Symp.\ on
  Principles of Database Systems (PODS~2008)}, pages 23--32, 2008.

\bibitem{tCKo09}
B.~ten Cate and P.~G. Kolaitis.
\newblock Structural characterizations of schema-mapping languages.
\newblock In {\em Proc.\ of the 12th Int.\ Conf.\ on Database Theory
  (ICDT~2009)}, pages 63--72, 2009.

\bibitem{TrCP09}
Q.~T. Tran, C.-Y. Chan, and S.~Parthasarathy.
\newblock Query by output.
\newblock In {\em Proc.\ of the ACM SIGMOD Int.\ Conf.\ on Management of Data},
  pages 535--548, 2009.

\bibitem{TsSI96}
O.~G. Tsatalos, M.~H. Solomon, and Y.~E. Ioannidis.
\newblock The {GMAP}: {A} versatile tool for phyisical data independence.
\newblock {\em Very Large Database J.}, 5(2):101--118, 1996.

\bibitem{VaWo86}
M.~Y. Vardi and P.~Wolper.
\newblock Automata-theoretic techniques for modal logics of programs.
\newblock {\em J.\ of Computer and System Sciences}, 32:183--221, 1986.

\bibitem{ZhMe05}
Z.~Zhang and A.~Mendelzon.
\newblock Authorization views and conditional query containment.
\newblock In {\em Proc.\ of the 10th Int.\ Conf.\ on Database Theory
  (ICDT~2005)}, volume 3363 of {\em Lecture Notes in Computer Science}, pages
  259--273. Springer, 2005.

\end{thebibliography}
%%\bibliography{medium-string,krdb}

\end{document}